\documentclass[smallextended, envcountsame]{svjour3}
\smartqed
\usepackage{graphicx}

\usepackage{amssymb}
\usepackage{amsmath}
\usepackage{latexsym}
\usepackage[misc,geometry]{ifsym}

\newenvironment{Bullet}
{\begin{list}{}%
             {\setlength\labelsep{0pt}%
              \setlength\itemindent{0pt}%
              \setlength\leftmargin{15pt}%
              \setlength\labelwidth{15pt}%
              \setlength\topsep{3pt}%
              \setlength\parsep{3pt}%
              \setlength\itemsep{0pt}%
              }
\item[$\bullet$\hfill]}%
{\end{list}}

\begin{document}

\spnewtheorem{algorithm}{Algorithm}{\bf}{}
\spnewtheorem{example_newdef}{Example}{\bf}{}

\title{On Decoding of DVR-Based Linear Network Codes 
}

\author{Qifu (Tyler) Sun \and
Shuo-Yen Robert Li}  

\institute{Q.~T.~Sun (\Letter) \at
        Institute of Advanced Networking Technology and New Service, University of Science and Technology Beijing, Beijing, P. R. China \\
              \email{qfsun@ustb.edu.cn}        
              \and
              S.-Y.~R.~Li \at
              Key Laboratory of Network Coding Key Technology and Application, Shenzhen, Shenzhen Research Institute, The Chinese University of Hong Kong, P. R. China \\
              \email{bobli@ie.cuhk.edu.hk} \\
}


\maketitle

\begin{abstract}
The conventional theory of linear network coding (LNC) is only over acyclic networks. Convolutional network coding (CNC) applies to all networks. It is also a form of LNC, but the linearity is w.r.t. the ring of rational power series rather than the field of data symbols. CNC has been generalized to LNC w.r.t. any discrete valuation ring (DVR) in order for flexibility in applications. For a causal DVR-based code, all possible source-generated messages form a free module, while incoming coding vectors to a receiver span the \emph{received submodule}. An existing \emph{time-invariant decoding} algorithm is at a delay equal to the largest valuation among all invariant factors of the received submodule. This intrinsic algebraic attribute is herein proved to be the optimal decoding delay. Meanwhile, \emph{time-variant decoding} is formulated. The meaning of time-invariant decoding delay gets a new interpretation through being a special case of the time-variant counterpart. The optimal delay turns out to be the same for time-variant decoding, but the decoding algorithm is more flexible in terms of decodability check and decoding matrix design. All results apply, in particular, to CNC.
\end{abstract}

\section{Introduction}\label{Sec:Introduction}
The fundamental theory of linear network coding (LNC) studies data propagation through an acyclic multicast network (\cite{LiYeungCai_03,Koetter_Medard_03}). The acyclic topology keeps data flowing from the upstream to the downstream. As the source pipelines messages into the network, the theory deals with each individual message separately by assuming appropriate buffering and synchronization. Meanwhile, the issue of data propagation delay can be set aside by assuming a delay-free network.

Without the acyclic assumption, on the other hand, data in sequential messages may convolve together through cyclic transmission. One way to deal with cyclic data propagation is by convolutional network coding (CNC) (\cite{Koetter_Medard_03,LiYeung_CNC_06}). To ensure causality, the data propagation delay must be nonzero around every cycle in the network. As proved in \cite{LiYeung_CNC_06}, causal data propagation can achieve simultaneous optimal data rate from the source to every other node, that is, at the rate equal to the maxflow from the source to each node. It is not sensible to assume a delay-free network with cycles. The data propagation delay is an essential factor in CNC. That includes \emph{decoding delay} as well as other forms of delay.

The symbol alphabet in LNC is algebraically structured as a finite field $\mathbb{F}$. Represent a message by an $\omega$-dim row vector over $\mathbb{F}$. A linear network code assigns a \emph{coding coefficient} from the symbol field $\mathbb{F}$ to every \emph{adjacent pair} of channels. Calculating from the upstream to the downstream, the coding coefficients naturally derive a \emph{coding vector} on every channel, which is an $\omega$-dim column vector over $\mathbb{F}$, such that the symbol transmitted on a channel is the dot product between the message and the coding vector. In CNC, on the other hand, the data unit is a pipeline of symbols, so is a coding coefficient. As explained in \cite{LiYeung_CNC_06}, these pipelines should be regarded as \emph{rational power series} over the symbol field $\mathbb{F}$ rather than polynomials or power series over $\mathbb{F}$. Thus, as a form of LNC, CNC is linear with respect to the ring $\mathbb{F}[(D)]$ of rational power series rather than to the field $\mathbb{F}$. Here the dummy variable $D$ stands for a unit-time delay.

While theory of LNC over the symbol field $\mathbb{F}$ applies only to acyclic networks, theory of CNC over the ring $\mathbb{F}[(D)]$ works on all networks. The underpinning reason is that the algebraic structure of the ring $\mathbb{F}[(D)]$ includes an acyclic attribute, namely, time that breaks the deadlock in cyclic transmission. This characteristic of $\mathbb{F}[(D)]$ is shared by every discrete valuation ring (DVR), which means a \emph{local} principal ideal domain (PID). For example, $p$-adic integers also form a DVR.

Every element in a DVR is equal to a power of the uniformizer subject to a unit factor, and the exponent is the (non-Archimedean) valuation of the element. All ideals in a DVR form an infinite chain under inclusion that mimics the unidirectional characteristic of time. The uniformizer in a DVR generalizes the role that the unit-time delay $D$ plays in $\mathbb{F}[(D)]$. The LNC theory in \cite{LiSun_11} is linear with respect to a general DVR over a cyclic network. We shall call this DVR-based LNC. It generalizes the notion of CNC and may potentially compensate for deficiencies of CNC, such as the difficulty in long-distance synchronization.

In DVR-based LNC, the total data generated by the source is represented by an $\omega$-dim vector over the DVR. The ensemble of all possible data units makes the \emph{source module}, which is an $\omega$-dim free module over the DVR. The incoming coding vectors to a node span a submodule to be called the \emph{received submodule} at the node. As a DVR is a PID, the received submodule is also a free module by the \emph{theorem of invariant factor decomposition of a free module over a PID} (See \cite{Dummit_Foote} for example). When the received submodule is of the full rank $\omega$, it differs from the source module by the invariant factors. In that case, the node is able to \emph{decode} the source data in the sense quoted by Definition 1 below. The \emph{decoding delay} is formulated in the form of an exponent of the uniformizer. In the case of CNC, that is, when the DVR is $\mathbb{F}[(D)]$, this is an amount of unit times. The decoding mechanism provided in \cite{LiSun_11} incurs a delay equal to the largest valuation among all invariant factors of the received submodule.

Decoding in the sense of Definition 1, however, involves a decoding matrix that depends on the identity of the whole received submodule. This leaves some issues to be further explicated, including:
\begin{Bullet}
What does the decoding delay at a receiving node mean?
\end{Bullet}
\begin{Bullet}
How to implement a decoding mechanism at a prescribed decoding delay?
\end{Bullet}
\begin{Bullet}
What is the optimal decoding delay and how to determine it?
\end{Bullet}

The present paper attempts to strengthen the study on the decoding of DVR-based LNC, and to answer the three questions raised in the above particularly. After a brief review of the said decoding mechanism, Theorem 3 in Section 2 asserts that the optimal decoding delay under this decoding mechanism is exactly the largest valuation among all invariant factors of the received submodule. Section 3 expresses elements in the DVR in the series form with the dummy variable being the uniformizer and the coefficients being coset representatives of the DVR over its unique maximal ideal. The arithmetic over such expressions is also described. Based on this power series expression, Section 4 formulates the notion of \emph{time-variant decoding}. Here the word ``time'' refers to the exponent of the uniformizer, which generalizes the role of time in the combined space-time domain represented by $\mathbb{F}[(D)]$. Decoding in the sense of Definition 1 will hereafter be referred to as \emph{time-invariant decoding}. By showing time-invariant decoding as a special case of time-variant decoding, the meaning of time-invariant decoding delay gets clarified. Theorem 7 proves that the optimal delay is the same for both decoding mechanisms, but the implementation of time-variant decoding is less constricted than time-invariant decoding in the sense that the check of decodability with a prescribed delay and all subsequent decoding matrices involved depend only on finitely many terms in the power series expression of coding vectors. Such a new decoding scheme for DVR-based LNC is proposed in Algorithm 2.

\section{Optimal Decoding Delay of DVR-Based LNC}\label{Sec:Optimal_Decoding_Delay}
Let every edge in a multicast network represent a channel of unit capacity for noiseless transmission. Multiple edges between nodes are allowed. There are $\omega$ outgoing channels from the source, which are called source channels.

Denote by $\mathbb{D}$ a DVR with the uniformizer $z$. A $\mathbb{D}$-linear network code ($\mathbb{D}$-LNC) assigns a coding coefficient $k_{d,e} \in \mathbb{D}$ to every adjacent pair $(d, e)$ of channels. The source generates $\omega$ data units belonging to $\mathbb{D}$, one to be sent out from each source channel. The $\omega$ data units are represented by an $\omega$-dim row vector $\mathbf{m}$. A $\mathbb{D}$-LNC is said to be \emph{causal} when, around every cycle in the network, there is at least one adjacent pair $(d, e)$ with $k_{d,e}$ divisible by $z$. Causality guarantees the existence of a unique set of coding vectors, which is to assign a vector $\mathbf{f}_e$ over $\mathbb{D}$ to each channel $e$ so that the channel carries the data unit calculable by $\mathbf{m}\cdot \mathbf{f}_e$. In particular, the coding vectors for the $\omega$ source channels form the natural basis of the free module $\mathbb{D}^\omega$, and the coding vector $\mathbf{f}_e$ for every outgoing channel $e$ of a non-source node $v$ is equal to $\sum_{d\in In(v)}k_{d,e}\mathbf{f}_d$.

Note that a DVR-based LNC in general does not necessarily correspond to a set of coding vectors and, when it does, the set may not be unique (See \cite{LiSun_11}). Hereafter all $\mathbb{D}$-LNCs considered in this paper are assumed to be causal.

The data unit received by a node $v$ from an incoming channel $e$ via a $\mathbb{D}$-LNC is $\mathbf{m}\cdot\mathbf{f}_e$. Denote by $In(v)$ the set of incoming channels to $v$. Juxtapose the incoming coding vectors to $v$ into the matrix $[\mathbf{f}_e]_{e\in In(v)}$. Thus, $\mathbf{m}\cdot[\mathbf{f}_e]_{e\in In(v)}$ is the row vector of all $|In(v)|$ data units received by $v$. Let $\mathbf{I}$ denote the $\omega \times \omega$ identity matrix. As a prerequisite for decoding at the node $v$, the matrix $[\mathbf{f}_e]_{e\in In(v)}$ over $\mathbb{D}$ must attain the full rank $\omega$. Assume that this is the case.

\begin{definition}(\cite{LiSun_11})
For a causal $\mathbb{D}$-LNC, a decoding matrix at a node $v$ with \emph{decoding delay} $\delta \geq 0$ is an $|In(v)|\times \omega$ matrix $\mathbf{A}$ over $\mathbb{D}$ such that
\begin{equation}\label{eqn:1}
[\mathbf{f}_e]_{e\in In(v)} \cdot \mathbf{A} = z^\delta \mathbf{I}
\end{equation}
Note that, when the matrix $\mathbf{A}$ exists, it is uniquely determined by $[\mathbf{f}_e]_{e\in In(v)}$.
\end{definition}

The meaning of decoding delay in this definition will be clarified when the decoding mechanism so formulated is shown as a special case of ``time-variant decoding'' in Section \ref{Sec:Time_vriant_decoding}.

The incoming coding vectors to the node $v$ generates a submodule $\langle\mathbf{f}_e: e\in In(v)\rangle$ of the free module $\mathbb{D}^\omega$ over $\mathbb{D}$. According to the \emph{theorem of invariant factor decomposition of a free module over a PID}, the module $\mathbb{D}^\omega$ is free of the rank $\omega$.
Moreover, there is a basis $\left\{\mathbf{u}_1, \cdots, \mathbf{u}_\omega\right\}$ of $\mathbb{D}^\omega$ and nonnegative integers $i_1 \leq \cdots \leq i_\omega$ such that $z^{i_1}\mathbf{u}_1, \cdots, z^{i_\omega}\mathbf{u}_\omega$ form a basis of $\langle\mathbf{f}_e: e \in In(v)\rangle$. Here the invariant factors refer to $z^{i_1}, \cdots, z^{i_\omega}$, of which the evaluations are $i_1, \cdots, i_\omega$, respectively.

\begin{lemma}\emph{(\cite{LiSun_11})}
Let the matrix $[\mathbf{f}_e]_{e\in In(v)}$ over $\mathbb{D}$ attain the full rank $\omega$ at a node $v$ for a causal $\mathbb{D}$-LNC. Then, a time-invariant decoding matrix exists with a decoding delay no more than the largest valuation among invariant factors of the submodule $\langle \mathbf{f}_e: e \in In(v)\rangle$ in $\mathbb{D}^\omega$.
\end{lemma}

The largest valuation in Lemma 2 turns out to be the exact minimum decoding delay:

\begin{theorem}
When a causal $\mathbb{D}$-LNC is decodable at node $v$, the minimum decoding delay is equal to the highest valuation among invariant factors of the submodule $\langle \mathbf{f}_e: e \in In(v)\rangle$ in $\mathbb{D}^\omega$.
\end{theorem}
\begin{proof}
According to the theorem of invariant factor decomposition of a free module over a PID, there is a basis $\left\{\mathbf{u}_1, \cdots, \mathbf{u}_\omega\right\}$ of $\mathbb{D}^\omega$ and nonnegative integers $i_1 \leq \cdots \leq i_\omega$ such that $z^{i_1}\mathbf{u}_1, \cdots, z^{i_\omega}\mathbf{u}_\omega$ form a basis of $\langle \mathbf{f}_e: e \in In(v)\rangle$ in $\mathbb{D}^\omega$.

Lemma 2 has shown that there exists an $|In(v)|\times \omega$ matrix $\mathbf{A}$ over $\mathbb{D}$ such that $[\mathbf{f}_e]_{e\in In(v)} \cdot \mathbf{A} = z^{i_\omega} \mathbf{I}$. Now let $\mathbf{A}$ be an arbitrary $|In(v)|\times \omega$ matrix over $\mathbb{D}$ such that $[\mathbf{f}_e]_{e\in In(v)} \cdot \mathbf{A} = z^i \mathbf{I}$ for some $i \geq 0$. It suffices to show that $i \geq i_\omega$. Because every $\mathbf{f}_e, e\in In(v)$, is generated by $z^{i_1}\mathbf{u}_1, \cdots, z^{i_\omega}\mathbf{u}_\omega$, there exists an $\omega \times |In(v)|$ matrix $\mathbf{M}$ over $\mathbb{D}$ such that
\[
[z^{i_j}\mathbf{u}_j]_{1\leq j \leq \omega} \cdot \mathbf{M} = [\mathbf{f}_e]_{e\in In(v)}
\]
Thus,
\[
~[\mathbf{u}_j]_{1\leq j \leq \omega} \cdot \left[\begin{matrix} z^{i_1} & & \mathbf{0} \\ & \ddots & \\ \mathbf{0} & & z^{i_\omega} \end{matrix} \right] \cdot \mathbf{M} \cdot \mathbf{A}
= [z^{i_j}\mathbf{u}_j]_{1\leq j \leq \omega} \cdot \mathbf{M} \cdot \mathbf{A}
= [\mathbf{f}_e]_{e\in In(v)} \cdot \mathbf{A}
= z^i \mathbf{I}
\]
where a bolded 0 stands for a cluster of zero entries. Consequently,
\[
\left[\begin{matrix} z^{i_1} & & \mathbf{0} \\ & \ddots & \\ \mathbf{0} & & z^{i_\omega} \end{matrix} \right]\cdot\mathbf{M}\cdot\mathbf{A} = z^i [\mathbf{u}_j]_{1\leq j \leq \omega}^{-1},
\]
and then
\[
\left[\begin{matrix} z^{i_1} & & \mathbf{0} \\ & \ddots & \\ \mathbf{0} & & z^{i_\omega} \end{matrix} \right]\cdot\mathbf{M}\cdot\mathbf{A}\cdot[\mathbf{u}_j]_{1\leq j \leq \omega} = z^i \mathbf{I}
\]
All entries in the last row of the product matrix on the left-hand side are divisible by $z^{i_\omega}$. Thus $i \geq i_\omega$. \hfill $\blacksquare$
\end{proof}

The invariant factors of the submodule $\langle \mathbf{f}_e: e \in In(v)\rangle$ in $\mathbb{D}^\omega$ can be calculated as follows (See \cite{Brown_Book} for example.) When $j$ is less than or equal to the rank of $\langle \mathbf{f}_e: e\in In(v)\rangle$, let $\Delta_j$ denote the greatest common divisor, up to a unit factor, of the determinants of all $j\times j$ submatrices in $[\mathbf{f}_e]_{e\in In(v)}$. Then, the invariant factors are $\Delta_1$, $\Delta_2/\Delta_1$, $\Delta_3/\Delta_2, \Delta_4/\Delta_3, \cdots$.

\begin{corollary}
When a causal $\mathbb{D}$-LNC is decodable at node $v$, the minimum decoding delay is equal to the valuation of $\Delta_\omega/\Delta_{\omega-1}$.
\end{corollary}

The special case of this corollary for $\mathbb{D} = \mathbb{F}[(D)]$ has also been deduced in \cite{Guo_Cai_IEICE}.

\begin{example_newdef}
Let $p = 3$ and $\mathbb{Z}_{(3)}$ denote the ring of rational $p$-adic integers, that is, rational numbers with denominators not divisible by 3. This ring qualifies as a DVR with the uniformizer $z = 3$. Fig. 1 prescribes the coding coefficients of a causal $\mathbb{Z}_{(3)}$-LNC over a cyclic network with $\omega = 2$ outgoing channels from the source. The figure also shows the coding vectors. For the particular node labeled by $v$, the matrix $[\mathbf{f}_e]_{e\in In(v)} = \left[ \begin{matrix} 2 & -z \\ 0 & -z \end{matrix}\right]$ over the DVR $\mathbb{Z}_{(3)}$. Thus, $\Delta_1 = 1, \Delta_2 = z$, and $\Delta_\omega/\Delta_{\omega-1} = z$ up to unit factors. By Corollary 4, the LNC is decodable at $v$ with the minimum delay 1. In this case the unique decoding matrix with delay 1 is $\left[ \begin{matrix} z/2 & ~-z/2 \\ 0 & -1 \end{matrix} \right]$. \hfill $\blacksquare$
\end{example_newdef}

\begin{figure}[htbp]
\centering
\scalebox{.9}
{\includegraphics{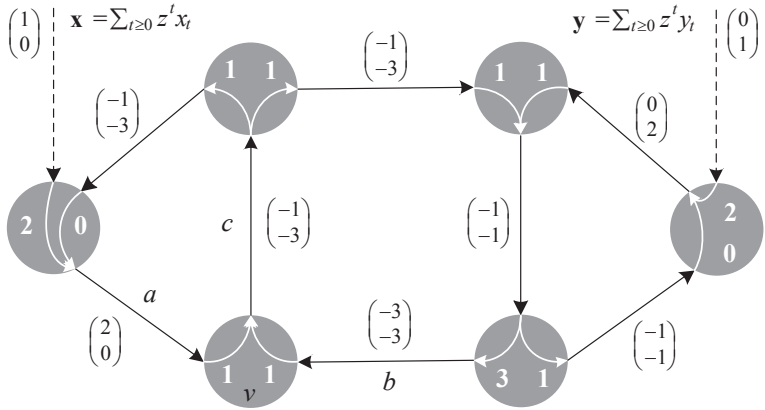}}
\caption{Let $p = 3$ and $\mathbb{Z}_{(3)}$ denote the ring of rational $p$-adic integers. This ring qualifies as a DVR with the uniformizer $z = 3$. The \emph{Shuttle Network} is a cyclic network with $\omega = 2$ outgoing channels from the source. A causal $\mathbb{Z}_{(3)}$-LNC is prescribed by coding coefficients and coding vectors. For the particular node labeled by $v$, the incoming channels are $a$ and $b$ and the outgoing channel is $c$. Example 1 asserts the decodability of the LNC at $v$ with the minimum delay 1. Example 2 in Section 3 expresses the data unit transmitted over channel $c$ as a power series over a set of coset representatives of $\mathbb{Z}_{(3)}/3\mathbb{Z}_{(3)}$. Example 3 in Section 4 illustrates the procedure of time-variant decoding at the node $v$.}
\end{figure}

\section{Power Series Expression of DVR-Based LNC}\label{Sec:Power_Series_Expression}
Let $R \subset \mathbb{D}$ be a complete set of coset representatives of $\mathbb{D}$ over the ideal $z\mathbb{D}$ and assume that $0 \in R$. Thus all nonzero elements of $R$ are units in $\mathbb{D}$. There can be many choices for such a set $R$. In the special case when $\mathbb{D} = \mathbb{F}[(D)]$, we simply fix $R = \mathbb{F}$. In general, $R$ is not closed under addition and multiplication.

\begin{proposition}
Every element of $\mathbb{D}$ can be uniquely expressed as an infinite series $\sum_{t\geq 0} a_tz^t$ with coefficients $a_t \in R$.
\end{proposition}
\begin{proof}
An arbitrary element $x_0$ of $\mathbb{D}$ can be expressed in the series form as follows. First, let $a_0 \in R$ be the coset representative of $x_0 + z\mathbb{D}$ and write $x_1 = (x_0-a_0)/z$. Then, iteratively for all $t\geq 1$, let $a_t \in R$ be the coset representative of $x_t + z\mathbb{D}$ and write $x_{t+1} = (x_t-a_t)/z$. By induction on $k$, we find $x_0 = \sum_{t\geq 0} a_tz^t$ modulo $z^k$ for all $k \geq 0$. Hence, $x_0-\sum_{t\geq 0} a_tz^t \in \bigcap_{t>0}z^t\mathbb{D}$. On the other hand, it can be derived, for example from \emph{Nakayama¡¯s lemma} in commutative algebra (See \cite{Brown_Book} for example,) that $\bigcap_{t>0}z^t\mathbb{D} = \{0\}$. Thus $x_0 = \sum_{t\geq 0}a_tz^t$.

Next we show the uniqueness of the power series expression. Suppose that $\sum_{t\geq 0}a_tz^t = \sum_{t\geq 0}b_tz^t$ with $a_t, b_t \in R$. Since $\sum_{t\geq 0}a_tz^t$ belongs to the coset $a_0 + z\mathbb{D}$ and $\sum_{t\geq 0}b_tz^t$ to the coset $b_0 + z\mathbb{D}$, both $a_0$ and $b_0$ represent the same coset and hence are equal. By induction on $t$, we then find $a_t = b_t$ for all $t$. \hfill $\blacksquare$
\end{proof}

Applying Proposition 5 to every entry in an arbitrary matrix $\mathbf{M}$ over $\mathbb{D}$, we can also express $\mathbf{M}$ in a unique way as $\sum_{j\geq 0} z^jM_j$, where every $M_j$ is a matrix over $R$. In particular when $\mathbf{M}$ is the row vector $\mathbf{m}$ of data units generated by the source, we adopt the following notation:
\begin{Bullet}
$\mathbf{m} = \sum_{j\geq 0} z^j m_j$
\end{Bullet}

\noindent In CNC, the source pipelines data in the time-divisioning manner. In every unit time, a \emph{message} consisting of $\omega$ symbols is produced. The above power series expression of $\mathbf{m}$ is a generalized form of such time-divisioning with each $m_j$ generalizing a message in CNC.

Similarly for a causal $\mathbb{D}$-LNC, we write:
\begin{Bullet}
$[\mathbf{f}_e]_{e\in In(v)} = \sum_{j\geq 0}z^jF_j$
\end{Bullet}
\begin{Bullet}
The $|In(v)|$-dim row vector $\mathbf{r}$ of incoming data units to $v = \sum_{j\geq 0}z^jr_j$
\end{Bullet}

\noindent The equation $\mathbf{r} = \mathbf{m}\cdot[\mathbf{f}_e]_{e\in In(v)}$ now becomes
\[\sum_{j=0}^{\infty} z^jr_j = \left(\sum_{j=0}^\infty z^jm_j\right)\cdot\left(\sum_{j=0}^\infty z^jF_j\right)\]
\noindent Or, equivalently,
\begin{equation}
\sum_{j=0}^t z^jr_j = \left(\sum_{j=0}^t z^jm_j\right)\cdot\left(\sum_{j=0}^t z^jF_j\right) \mod z^{t+1} ~~~\mathrm{for}~~~t\geq 0
\end{equation}

Even though coset representatives are not closed under addition and multiplication, the arithmetic for power series over $R$ can be implemented by the following algorithm.

\vspace{6pt}

\begin{algorithm}
Denote by $\sigma$ the natural mapping from $\mathbb{D}$ onto $R$ such that $\sigma(a) = a$ modulo $z$ for all $a \in \mathbb{D}$. Then, the sum $\sum_{t\geq 0}a_tz^t + \sum_{t\geq 0}b_tz^t$ can be expressed as a power series $\sum_{t\geq 0}c_tz^t$ over $R$, where the coefficients $c_t$ are iteratively calculated together with a sequence $\{r_t\}_{t\geq 0}$ over $\mathbb{D}$ as follows.
\begin{Bullet}
$c_t = \sigma(a_t + b_t + r_t)$, where $r_0 = 0$
\end{Bullet}
\begin{Bullet}
$r_{t+1} = (a_t + b_t + r_t - c_t)/z$
\end{Bullet}
Meanwhile, the product $\left(\sum_{t\geq0}a_t z^t\right)\cdot\left(\sum_{t\geq 0}b_tz^t\right)$ can be expressed as a power series $\sum_{t\geq 0} d_tz^t$ over $R$, where the coefficients $d_t$ are iteratively calculated together with a sequence $\{s_t\}_{t\geq 0}$ over $\mathbb{D}$ as follows.
\begin{Bullet}
$d_t = \sigma\left(\sum_{j=0}^ta_jb_{t-j} + s_t\right)$, where $s_0 = 0$
\end{Bullet}
\begin{Bullet}
$s_{t+1} = \left(\sum_{j=0}^t a_j b_{t-j} + s_t - d_t\right)/z$ \hfill $\blacksquare$
\end{Bullet}
\end{algorithm}

\vspace{6pt}

\begin{example_newdef}
Consider the same causal $\mathbb{Z}_{(3)}$-LNC described in Example 1 and Fig. 1. A complete set of coset representatives of $\mathbb{Z}_{(3)}$ over the ideal $3\mathbb{Z}_{(3)}$ is $R = \{0, 1, 2\}$. With respect to $R$, the power series expression of the coding vectors for channels $a$, $b$ and $c$ are, respectively,
\[
\mathbf{f}_a = \left(\begin{matrix} 2 \\ 0 \end{matrix}\right),
\mathbf{f}_b = z\left(\begin{matrix} 2 \\ 2 \end{matrix}\right)+z^2\left(\begin{matrix} 2 \\ 2 \end{matrix}\right)+\cdots,
\mathbf{f}_c = \left(\begin{matrix} 2 \\ 0 \end{matrix}\right)+z\left(\begin{matrix} 2 \\ 2 \end{matrix}\right)+z^2\left(\begin{matrix} 2 \\ 2 \end{matrix}\right)+\cdots
\]

Let the source generate a power series of messages $\mathbf{m} = (\sum_{t\geq 0}z^tx_t, \sum_{t\geq 0}z^ty_t)$ over the symbol alphabet $R$. The data unit transmitted over a channel can be expressed as a power series over $R$. Take the channel $c$ as the example. We want to express the data unit $\mathbf{m}\cdot\mathbf{f}_c$ as a power series $\sum_{t\geq 0}z^td_t$ over $R$. Let $\sigma$ be the operation of modulo 3. Following Algorithm 1, the coefficients $d_t$ and auxiliary parameters $s_t$ are calculated iteratively as:
\begin{Bullet}
$s_0 = 0$ and $d_0 = \sigma(2x_0)$
\end{Bullet}
\begin{Bullet}
$s_1 = (2x_0 - d_0)/z$ and $d_1 = \sigma(2x_0 + 2y_0 + 2x_1 + s_1)$
\end{Bullet}
\begin{Bullet}
$s_2 = (2x_0 + 2y_0 + 2x_1 + s_1 - d_1)/z$ and $d_2 = \sigma(2x_0+2y_0+2x_1+2y_1+2x_2+s_2)$
\end{Bullet}
\begin{Bullet}
And so on. \hfill $\blacksquare$
\end{Bullet}
\end{example_newdef}

\section{Time-Variant Decoding of DVR-Based LNC}\label{Sec:Time_vriant_decoding}
Hereafter decoding in the sense of Definition 1 will be referred to as \emph{time-invariant decoding}. The power series expression of DVR-based LNC gives rise to a more general way to formulate the notion of decoding.

\begin{definition}
A causal $\mathbb{D}$-LNC is \emph{time-variant decodable} at a node $v$ with \emph{delay} $\delta$ when every $m_t$ can be $\mathbb{D}$-linearly calculated from $r_0, \cdots, r_{t+\delta}$ and $F_0, \cdots, F_{t+\delta}$. More explicitly, for all $t\geq 0$, there are $|In(v)|\times \omega$ matrices $A_{t,0}, \cdots, A_{t,t+\delta}$ over $R$ that are derivable from $F_0, \cdots, F_{t+\delta}$ such that
\[
\sum_{j=0}^t z^{j+\delta}m_j = \left(\sum_{j=0}^{t+\delta} z^jr_j\right)\cdot\left(\sum_{j=0}^{t+\delta} z^jA_{t,j}\right) \mod z^{t+\delta+1}
\]
\end{definition}

The word ``time'' in this definition refers to exponent of $z$, which generalizes the unit-time delay $D$ in CNC. As explained in the proof of the next theorem, time-invariant decoding with delay $\delta$ qualifies as a special case of time-variant decoding with the same delay. This offers a new interpretation of the time-invariant decoding delay. The notion of time-variant decoding has earlier been formulated in \cite{Guo_Cai_Sun} for CNC, which now coincides with Definition 6 for the special case when $\mathbb{D} = \mathbb{F}[(D)]$. The equivalence between items (b) and (c) in the next theorem has also been given in \cite{Guo_Cai_Sun} for the special case of CNC.

\begin{theorem}
For a causal $\mathbb{D}$-LNC, the following statements are equivalent at every node $v$:

$\mathrm{(a)}$ Time-invariant decodability with delay $\delta$

$\mathrm{(b)}$ Time-variant decodability with delay $\delta$

$\mathrm{(c)}$ $\mathbb{D}$-linear calculation of $m_0$ from $r_0, \cdots, r_\delta$ and $F_0,\cdots, F_\delta$
\end{theorem}
\begin{proof}
The definition of time-variant decoding directly implies (b) $\Rightarrow$ (c).

To prove (a) $\Rightarrow$ (b), we shall show that time-invariant decoding with delay $\delta$ is a special case of time-variant decoding with the same delay. Let the matrix $\mathbf{A}$ over $\mathbb{D}$ be such that $(\sum_{j\geq 0} z^j F_j)\cdot\mathbf{A} = z^\delta \mathbf{I}$. Writing $\mathbf{A} = \sum_{j\geq 0} z^jA_j$, where the matrices $A_j$ are over $R$, we have, for all $t \geq 0$,
\[
z^\delta \mathbf{I} = \left( \sum_{j=0}^{t+\delta} z^jF_j \right)\cdot\left( \sum_{j=0}^{t+\delta}z^jA_j \right) \mod z^{t+\delta+1}
\]
Consequently, for all $t \geq 0$,
\begin{equation*}
\begin{aligned}
&~~~\left( \sum_{j=0}^{t+\delta}z^jr_j \right)\cdot\left( \sum_{j=0}^{t+\delta}z^jA_j \right) \\
&= \left( \sum_{j=0}^{t+\delta}z^jm_j \right)\cdot\left( \sum_{j=0}^{t+\delta}z^jF_j \right)\cdot\left( \sum_{j=0}^{t+\delta}z^jA_j \right)\\
&= z^\delta\left( \sum_{j=0}^{t+\delta}z^jm_j \right) \\
&= \sum_{j=0}^t z^{j+\delta}m_j \mod z^{t+\delta+1}
\end{aligned}
\end{equation*}

As the matrix $\mathbf{A}$ is calculable from $[\mathbf{f}_e]_{e\in In(v)} = \sum_{j\geq 0}z^jF_j$, each $A_j$ is calculable from $F_0, \cdots, F_{j+\delta}$. Thus the $\mathbb{D}$-LNC is time-variant decodable with delay $\delta$.

It remains to prove (c) $\Rightarrow$ (a). Let $A_0, \cdots, A_\delta$ be $|In(v)|\times\omega$ matrices over $R$ that are derivable from $F_0, \cdots, F_\delta$ such that
\[
z^\delta m_0 = \left(\sum_{j=0}^\delta z^jr_j\right)\cdot\left(\sum_{j=0}^\delta z^jA_{t,j}\right) \mod z^{\delta+1}
\]
Substituting (2) into the above equation,
\[
z^\delta m_0 = \left(\sum_{j=0}^\delta z^jm_j\right)\cdot\left(\sum_{j=0}^\delta z^jF_j\right)\cdot\left(\sum_{j=0}^\delta z^jA_{0,j}\right) \mod z^{\delta+1}
\]
As the above equation holds for all possible $m_0, m_1, \cdots, m_\delta$,
\[
z^\delta m_0 = m_0\cdot \left(\sum_{j=0}^\delta z^jF_j\right)\cdot\left(\sum_{j=0}^\delta z^jA_{0,j}\right) \mod z^{\delta+1}
\]
and hence
\[
z^\delta\mathbf{I} = \left(\sum_{j=0}^\delta z^jF_j\right)\cdot\left(\sum_{j=0}^\delta z^jA_{j}\right) \mod z^{\delta+1}
\]
Then,
\[
z^\delta\mathbf{I} = \left(\sum_{j=0}^\infty z^jF_j\right)\cdot\left(\sum_{j=0}^\delta z^jA_{j}\right) \mod z^{\delta+1}
\]
Thus there exists an $\omega\times\omega$ matrix $\mathbf{M}$ over $\mathbb{D}$ such that
\[
z^\delta(\mathbf{I} + z\mathbf{M}) = \left(\sum_{j=0}^\infty z^j F_j \right)\cdot\left(\sum_{j=0}^\delta z^j A_j\right)
\]
Because $\det(\mathbf{I} + z\mathbf{M}) = 1$ modulo $z$, the matrix $\mathbf{I} + z\mathbf{M}$ is invertible over $\mathbb{D}$. Thus
\[
z^\delta\mathbf{I} = \left(\sum_{j=0}^\infty z^j F_j \right)\cdot\left(\sum_{j=0}^\delta z^j A_j\right)\cdot(\mathbf{I} + z\mathbf{M})^{-1}
\]
This establishes $\left(\sum_{j=0}^\delta z^j A_j\right)\cdot(\mathbf{I} + z\mathbf{M})^{-1}$ as the desired time-invariant decoding matrix. \hfill $\blacksquare$
\end{proof}

Because of Theorem 7, there is no distinction between time-invariant decodability and time-variant decodability. However, the distinction still remains in decoding matrices and in decoding algorithms. Assuming $\mathbb{D}$-linear calculation of $m_0$ from $r_0, \cdots, r_\delta$ and $F_0, \cdots, F_\delta$, the proof of ``(c) $\Rightarrow$ (a)'' in the above provides an algorithm for calculating the time-invariant decoding matrix. The calculation though involves the inversion of the matrix $\mathbf{I}+z\mathbf{M}$ over $\mathbb{D}$, which is in turn computed from $[\mathbf{f}_e]_{e\in In(v)}$, whose power series expression may contain infinite terms. This handicaps the implementation of the time-invariant decoding mechanism. The following algorithm first computes a time-variant decoding matrix $\mathbf{A}_0$ over $\mathbb{D}$ from $r_0, \cdots, r_\delta$ and $F_0, \cdots, F_\delta$, and then dynamically calculates $m_t$ from $r_0, \cdots, r_{t+\delta}$ and $F_0,\cdots, F_{t+\delta}$ for all $t > 0$ without computing the matrices $A_{t,0}, \cdots, A_{t,t+\delta}$ formulated in Definition 6.

\begin{algorithm}
Assume that a causal $\mathbb{D}$-LNC is decodable at a node $v$ with the delay $\delta$. The time-variant decoding matrices $A_0, \cdots, A_\delta$ over $R$ can be computed from $F_0, \cdots, F_\delta$ as follows. By applying the invariant factor decomposition algorithm (see Chapter 15 in \cite{Brown_Book} for example), calculate:
\begin{Bullet}
the invariant factors $z^{i_1}, \cdots, z^{i_\omega}$ of $\sum_{j=0}^\delta z^jF_j$,
\end{Bullet}
\begin{Bullet}
an $\omega\times\omega$ invertible matrix $\mathbf{U}$ over $\mathbb{D}$, and
\end{Bullet}
\begin{Bullet}
an $|In(v)|\times|In(v)|$ invertible matrix $\mathbf{V}$ over $\mathbb{D}$ such that
\[
\sum_{j=0}^\delta z^jF_j = \mathbf{U}\cdot\left[
\begin{matrix}
z^{i_1} & & \mathbf{0} & \mathbf{0} \\
        & \ddots & & \vdots \\
\mathbf{0} & & z^{i_\omega} & \mathbf{0}
\end{matrix}
\right]\cdot\mathbf{V}
\]
\end{Bullet}
Then, calculate
\[
\mathbf{A} = \mathbf{V}^{-1}\cdot
\left[\begin{matrix}
z^{\delta-i_1} & & \mathbf{0} \\
        & \ddots & \\
\mathbf{0} & & z^{\delta-i_\omega} \\
\mathbf{0} & \cdots & \mathbf{0} \end{matrix}
\right]\cdot \mathbf{U}^{-1}~~\mathrm{over}~\mathbb{D}
\]
Let $A_0, \cdots, A_\delta$ denote the first $\delta + 1$ matrix coefficients over $R$ in the power series expression of $\mathbf{A}$. They are the desired time-variant decoding matrices.

For brevity, denote by $\mathbf{A}_0$ the matrix $\sum_{j=0}^\delta z^jA_j$ over $\mathbb{D}$. Then, $m_0$ can be decoded from $r_0, \cdots, r_\delta$ by
\begin{equation}
m_0 = \sigma\left(\frac{\left(\sum_{j=0}^\delta z^j r_j\right)\cdot\mathbf{A}_0}{z^\delta}\right)
\end{equation}
where $\sigma$ is the mapping defined in Algorithm 1 and applies to a vector in a componentwise way.

For all $t \geq 1$, $m_t$ can be iteratively decoded from $r_0, \cdots, r_{t+\delta}$ and $F_0, \cdots, F_{t+\delta}$ by
\begin{equation}
m_t = \sigma\left(\frac{\left(\sum_{j=0}^{t+\delta}z^jr_j - \left(\sum_{j=0}^{t-1}z^jm_j\right)\cdot\left(\sum_{j=0}^{t+\delta}z^jF_j\right)\right)\cdot\mathbf{A}_0}{z^{\delta+t}}\right)
\end{equation}
\hfill $\blacksquare$
\end{algorithm}

\underline{\emph{Justification}}~~In the algorithm, the resultant $A_0, \cdots, A_\delta$ satisfy
\begin{equation*}
\begin{aligned}
&~~~\left(\sum_{j=0}^\delta z^j F_j\right)\cdot\left(\sum_{j=0}^\delta z^j A_j\right) \\
&= \left(\sum_{j=0}^\delta z^j F_j\right)\cdot\left(\sum_{j=0}^\infty z^j A_j\right) \\
&= \mathbf{U}\cdot
\left[\begin{matrix}
z^{i_1} & & \mathbf{0} & \\
 & \ddots & & \mathbf{0} \\
\mathbf{0} & & z^{i_\omega} &
\end{matrix}\right]
\cdot\mathbf{V} \cdot \mathbf{V}^{-1} \cdot
\left[\begin{matrix}
z^{\delta-i_1} & & \mathbf{0}\\
 & \ddots & \\
\mathbf{0} & & z^{\delta - i_\omega} \\
& \mathbf{0} &
\end{matrix}\right]
\cdot \mathbf{U}^{-1} \\
&= z^\delta\mathbf{I} \mod z^{\delta+1}
\end{aligned}
\end{equation*}
Thus
\begin{equation*}
\begin{aligned}
z^\delta m_0 &= m_0\cdot\left(\sum_{j=0}^\delta z^j F_j \right)\cdot\left(\sum_{j=0}^\delta z^j A_j \right) \\
& = \left(\sum_{j=0}^\delta z^j m_j \right)\cdot\left(\sum_{j=0}^\delta z^j F_j \right)\cdot\left(\sum_{j=0}^\delta z^j A_j \right) \\
& = \left(\sum_{j=0}^\delta z^j r_j \right)\cdot\left(\sum_{j=0}^\delta z^j A_j \right) \mod z^{\delta+1}
\end{aligned}
\end{equation*}
and consequently equation (3) holds.

In order to justify (4), observe that
\begin{equation*}
\begin{aligned}
&~~~~~\sum_{j=0}^{t+\delta} z^j r_j - \left(\sum_{j=0}^{t-1} z^j m_j \right)\cdot\left(\sum_{j=0}^{t+\delta} z^j F_j \right) \\
&= \left(\sum_{j=0}^{t+\delta} z^j m_j \right)\cdot\left(\sum_{j=0}^{t+\delta} z^j F_j \right) - \left(\sum_{j=0}^{t-1} z^j m_j \right)\cdot\left(\sum_{j=0}^{t+\delta} z^j F_j \right) \\
&= \left(\sum_{j=t}^{t+\delta} z^j m_j \right)\cdot\left(\sum_{j=0}^{t+\delta} z^j F_j \right) \\
&= z^t\left(\sum_{j=0}^{\delta} z^j m_{t+j} \right)\cdot\left(\sum_{j=0}^{t+\delta} z^j F_j \right) \mod z^{t+\delta+1}
\end{aligned}
\end{equation*}
and
\[
z^{t+\delta} \mathbf{I} = z^t\left(\sum_{j=0}^\delta z^j F_j \right)\cdot\mathbf{A}_0 = z^t\left(\sum_{j=0}^{t+\delta} z^j F_j \right)\cdot\mathbf{A}_0 \mod z^{t+\delta+1}
\]
Thus
\begin{equation*}
\begin{aligned}
&~~~~~\left(\sum_{j=0}^{t+\delta}z^jr_j - \left(\sum_{j=0}^{t-1}z^jm_j\right)\cdot\left(\sum_{j=0}^{t+\delta}z^jF_j\right)\right)\cdot\mathbf{A}_0 \\
&= z^t\left(\sum_{j=0}^{\delta}z^jm_{t+j}\right)\cdot\left(\sum_{j=0}^{t+\delta}z^jF_j\right)\cdot\mathbf{A}_0  \\
&= z^{t+\delta}\left(\sum_{j=0}^\delta z^j m_{t+j} \right) \\
&= z^{t+\delta}m_t \mod z^{t+\delta+1}
\end{aligned}
\end{equation*}
Consequently, (4) holds. \hfill $\blacksquare$

\vspace{12pt}

In Algorithm 2, every $m_t$, $t \geq 1$, is decoded by iterating the formula (4). The complexity of this decoding process is reduced in the following modified algorithm, which is based on the arithmetic over $R$ described in Algorithm 1.

\vspace{12pt}

\noindent {\bf Algorithm 2'}~~Follow the steps in Algorithm 2 till the formula (3) to establish the matrix $\mathbf{A}_0 = \sum_{j=0}^\delta z^jA_j$ and decode $m_0$. The next routine iteratively decodes $m_t$, $t \geq 1$, from $\mathbf{A}_0$ and the dynamically updated $r_t, \cdots, r_{t+\delta}$.

\vspace{6pt}

\indent For initialization, set $s_0 = s_1 = \cdots = s_{\delta+1} = 0$;

\indent For $t \geq 1$ do

\indent \{

\indent\indent $s_0 := \left(m_{t-1}\cdot F_0 - \sigma(m_{t-1}\cdot F_0)\right)/z$;

\indent\indent For $j \in [0, \delta-1]$ do

\indent\indent \{

\indent\indent\indent $r'_{t+j} := \sigma(r_{t+j} - m_{t-1}\cdot F_{j+1} - s_j)$;

\indent\indent\indent $s_{j+1} := (m_{t-1}\cdot F_{j+1} + s_j + r'_{t+j} - r_{t+j})/z$;

\indent\indent\indent $r_{t+j} := r'_{t+j}$;

\indent\indent \}

\indent\indent $j:= \delta$;

\indent\indent $r'_{t+j} := \sigma\left(r_{t+j} - \sum_{k=0}^{t-1}m_k\cdot F_{t+j-k} -s_j - s_{j+1} \right)$;

\indent\indent $s_{j+1} := \left(\sum_{k=0}^{t-1}m_k\cdot F_{t+j-k} + s_j + s_{j+1} + r'_{t+j} - r_{t+j} \right)/z$;

\indent\indent $r_{t+j} := r'_{t+j}$;

\indent\indent $m_t := \sigma\left(\frac{\left(\sum_{j=0}^\delta z^jr_{t+j}\right)\cdot\mathbf{A}_0}{z^\delta}\right)$;

\indent \} \hfill $\blacksquare$

\vspace{12pt}

The assumption in Algorithm 2, as well as in Algorithm 2', is the decodability of a causal $\mathbb{D}$-LNC at a node $v$ with the delay $\delta$. Such decodability can actually be determined based on the matrix $\sum_{j=0}^\delta z^jF_j$ instead of $\sum_{j=0}^\infty z^jF_j$ :

\begin{corollary}
A causal $\mathbb{D}$-LNC is decodable at node $v$ with delay $\delta$ if and only if
\emph{
\begin{description}
 \item[(5)] \emph{the matrix $\sum_{j=0}^\delta z^jF_j$ is of the full rank $\omega$ and}
 \item[(6)] \emph{$\delta$ is no smaller than the valuation of $\Delta_{\delta,\omega}/\Delta_{\delta,\omega-1}$, where $\Delta_{\delta,t}$ denotes the greatest common divisor of the determinants of all $t\times t$ submatrices in $\sum_{j=0}^\delta z^jF_j$.}
\end{description}
}
\end{corollary}
\begin{proof}
Let $z^{i_1}, \cdots, z^{i_\omega}$ be the invariant factors of $\sum_{j=0}^\delta z^j F_j$. According to the remark on invariant factors preceding Corollary 4, the conditions (5) and (6) are equivalent to:
\begin{description}
\item[(7)] $\sum_{j=0}^\delta z^jF_j$ is of the full rank $\omega$ and $\delta$ is not smaller than any among $i_1, \cdots, i_\omega$.
\end{description}

The statement (c) in Theorem 7 means that a matrix $\mathbf{A}_0$ over $\mathbb{D}$ can be derived from $F_0, \cdots, F_\delta$ such that
\[
z^\delta m_0 = \left(\sum_{j=0}^\delta z^j r_j \right)\cdot\mathbf{A}_0 \mod ~ z^{\delta+1}
\]
or equivalently,
\[
\tag{8}z^\delta\mathbf{I} = \left(\sum_{j=0}^\delta z^jF_j\right)\cdot\mathbf{A}_0 \mod~z^{\delta+1}
\]
It suffices to show the equivalence between this statement and (7). Denote by $\mathbf{U}$ and $\mathbf{V}$, respectively, the $\omega\times\omega$ invertible matrix and the $|In(v)|\times|In(v)|$ invertible matrix over $\mathbb{D}$ derived from $F_0, \cdots, F_\delta$ such that
\[
\sum_{j=0}^\delta z^j F_j = \mathbf{U}\cdot\left[\begin{matrix}z^{i_1} & & \mathbf{0} & \mathbf{0} \\
& \ddots & & \vdots \\ \mathbf{0} & & z^{i_\omega} & \mathbf{0} \end{matrix}\right]\cdot\mathbf{V}
\]
Thus, if (7) holds, then the matrix $\mathbf{A}_0 = \mathbf{V}^{-1} \cdot \left[\begin{matrix}
z^{\delta-i_1} & & \mathbf{0} \\ & \ddots & \\ \mathbf{0} & & z^{\delta-i_\omega} \\ \mathbf{0} & \cdots & \mathbf{0}
\end{matrix}\right]\cdot \mathbf{U}^{-1}$ will satisfy (8). On the contrary, if there is a matrix $\mathbf{A}_0$ over $\mathbb{D}$ subject to (8), then there exists a matrix $\mathbf{M}$ over $\mathbb{D}$ such that
\[
z^\delta\left(\mathbf{I} + z\mathbf{M}\right) = \left(\sum_{j=0}^\delta z^jF_j\right)\cdot\mathbf{A}_0
= \mathbf{U}\cdot\left[\begin{matrix}z^{i_1} & & \mathbf{0} & \mathbf{0} \\
& \ddots & & \vdots \\ \mathbf{0} & & z^{i_\omega} & \mathbf{0} \end{matrix}\right]\cdot\mathbf{V}\cdot\mathbf{A}_0
\]
Since both $\mathbf{U}$ and $\mathbf{I}+z\mathbf{M}$ are invertible over $\mathbb{D}$,
\[
z^\delta\mathbf{I} = \mathbf{U}\cdot\left[\begin{matrix}z^{i_1} & & \mathbf{0} & \mathbf{0} \\
& \ddots & & \vdots \\ \mathbf{0} & & z^{i_\omega} & \mathbf{0} \end{matrix}\right]\cdot\mathbf{V}\cdot \mathbf{A}_0 \cdot (\mathbf{I}+z\mathbf{M})^{-1}
= \left[\begin{matrix}z^{i_1} & & \mathbf{0} & \mathbf{0} \\
& \ddots & & \vdots \\ \mathbf{0} & & z^{i_\omega} & \mathbf{0} \end{matrix}\right]\cdot\mathbf{V}\cdot \mathbf{A}_0 \cdot (\mathbf{I}+z\mathbf{M})^{-1}\cdot\mathbf{U}^{-1}
\]
This implies that $\delta$ is not smaller than any among $i_1, \cdots, i_\omega$. \hfill $\blacksquare$
\end{proof}

\vspace{6 pt}
\begin{example_newdef}
Adopt the notation in Example 2. For the node $v$ with incoming channels $a$ and $b$, $\sum_{j=0}^\infty z^jF_j = \left[\begin{matrix} 2 & ~-z \\ 0 & ~-z \end{matrix}\right]$. In particular, $F_0 = \left[\begin{matrix} 2 & ~0 \\ 0 & ~0 \end{matrix}\right]$ and $F_1 = \left[\begin{matrix} 0 & ~2 \\ 0 & ~2 \end{matrix}\right]$. The matrix $F_0$ is not of full rank, and hence the code is not decodable at node $v$ with delay $0$. On the other hand, $F_0 + zF_1 = \left[\begin{matrix} 2 & ~2z \\ 0 & ~2z \end{matrix}\right]$ is of full rank. Since $\Delta_{1,2} = 4z$ and $\Delta_{1,1} = 2$, the valuation of $\Delta_{1, 2}/\Delta_{1, 1}$ is 1 and hence the $\mathbb{Z}_{(3)}$-LNC is decodable at node $v$ with delay $1$. From Algorithm 2, the matrix $F_0 + zF_1 = \left[\begin{matrix} 2 & ~2z \\ 0 & ~2z \end{matrix}\right]$ leads to a time-variant decoding matrix $\mathbf{A}_0 = \left[\begin{matrix} 2z & ~z \\ 0 & ~2+z \end{matrix}\right]$ over $\mathbb{D}$. In comparison, the time-invariant decoding matrix $\left[ \begin{matrix} z/2 & ~-z/2 \\ 0 & -1 \end{matrix} \right]$ at node $v$, as illustrated in Example 1, is computed based on $\sum_{j=0}^\infty z^jF_j$, and it has infinite nonzero terms in power series expression.

Assume that the source generates a sequence of messages $\mathbf{m} = (2~~1)+z(2~~1)+z^2(1~~2)+z^3(1~~1)$ over the symbol alphabet $R$. Then, for the node $v$, the power series of received symbol vectors over $R$ is
\[
\sum_{j=0}^\infty z^jr_j = (1~~0)+z(2~~2)+z^2(0~~2)+z^3(0~~1)+z^4(1~~2)+\cdots
\]
Based on the time-variant decoding matrix $\left[\begin{matrix} 2z & ~z \\ 0 & ~2+z \end{matrix}\right]$ over $\mathbb{D}$, $m_0$ can be decoded by formula (3) from $r_0$ and $r_1$:
\[
\sigma\left(\frac{(1~~0)\cdot\left[\begin{matrix} 2z & ~z \\ 0 & ~2+z \end{matrix}\right]}{z}\right) = (2 ~~ 1)
\]
Table I lists the dynamically updated parameters computed by the routine in Algorithm 2' for decoding $m_t$, $t \geq 1$. \hfill $\blacksquare$
\end{example_newdef}

\begin{table}
\caption{The dynamically updated parameters computed by the routine in Algorithm 2' for decoding $m_t$, $t \geq 1$ in Example 3.}
\center
\begin{tabular}{c|c|c|c}
  \hline
    & after $t = 1$ & after $t = 2$ & after $t=3$ \\ \hline
  $s_0$ & (1~~0) & (1~~0) & (0~~0) \\ \hline
  $r_t$ & (1~~0) & (2~~0) & (2~~0) \\ \hline
  $s_1$ & (0~~2) & (1~~2) & (0~~2) \\ \hline
  $r_{t+1}$ & (0~~0) & (2~~0) & (0~~1) \\ \hline
  $s_2$ & (0~~2) & (1~~5) & (0~~8) \\ \hline
  $m_t = \sigma((r_t + zr_{t+1})\cdot\mathbf{A}_0)$ & (2~~1) & (1~~2) & (1~~1) \\
  \hline
\end{tabular}
\end{table}

\section{Summary}\label{Sec:Conclusion}
The conventional theory of linear network coding (NC) deals with only acyclic networks. Convolutional network coding (CNC) applies to all networks. It is also a form of linear NC, but the linearity is w.r.t. the ring $\mathbb{F}[(D)]$ of rational power series instead of the field $\mathbb{F}$ of data symbols. The ring $\mathbb{F}[(D)]$ qualifies as a discrete valuation ring (DVR). CNC has previously been generalized to linear NC w.r.t. any DVR \cite{LiSun_11} for potential enhancement on applicability.

Some issues on DVR-based NC, including the special case of CNC, naturally arise: What is decoding delay at a receiving node? How to implement a decoding mechanism at a prescribed decoding delay? What is the optimal decoding delay and how to determine it?

Initially decoding of DVR-based NC at a node is defined in a \emph{time-invariant} manner in \cite{LiSun_11}, which also provides an algorithm to decode a causal DVR-based NC with a delay equal to the largest valuation among all invariant factors of the received submodule at the node, while the involved decoding matrix depends on the identity of the whole received submodule. Meanwhile, \emph{time-variant decoding} for CNC has been considered in \cite{Guo_Cai_Sun}. This notion though relies on the special characteristic of convolutional arithmetic over the particular DVR $\mathbb{F}[(D)]$.

The present paper attempts to strengthen the study on decoding a DVR-based NC, including the special case of CNC. First, the optimal delay in time-invariant decoding is shown to be exactly the aforementioned largest valuation. Then, by expressing elements in the DVR in the series form with the dummy variable being the uniformizer and the coefficients being coset representatives of the DVR over its unique maximal ideal, the notion of time-variant decoding is generalized from CNC to DVR-based NC. By showing time-invariant decoding as a special case of time-variant decoding, the meaning of time-invariant decoding delay formulated in \cite{LiSun_11} gets clarified. Although the optimal delay $\delta$ turns out to be the same for both time-invariant decoding and time-variant decoding, the latter is less constricted because both the check of the decodability with delay $\delta$ and the design of the involved decoding matrix depend only on the lowest $\delta+1$ terms in the power series expression of coding vectors. Such a decoding scheme is also proposed in this paper.




\end{document}